\documentclass[12pt, reqno]{amsart}

\usepackage{amsmath, amsthm, amscd, amsfonts, amssymb, graphicx, color}
\usepackage[bookmarksnumbered, colorlinks, plainpages]{hyperref}

\hypersetup{colorlinks=true,linkcolor=red, anchorcolor=green, citecolor=cyan, urlcolor=red, filecolor=magenta, pdftoolbar=true}

\usepackage{color}

\newtheorem{theorem}{Theorem}[section]
\newtheorem{lemma}[theorem]{Lemma}
\newtheorem{proposition}[theorem]{Proposition}
\newtheorem{corollary}[theorem]{Corollary}
\theoremstyle{definition}

\theoremstyle{remark}
\newtheorem{remark}[theorem]{Remark}

\newcommand{\uinorm}[1]{{\left\vert\kern-0.25ex\left\vert\kern-0.25ex\left\vert #1
 \right\vert\kern-0.25ex\right\vert\kern-0.25ex\right\vert}}

\begin{document}

\title[A new estimation of the quantum Chernoff bound]{A new estimation of the quantum Chernoff bound}

\author[M. Kian,T. H. Dinh, M. S. Moslehian, \MakeLowercase{and} H. Osaka]{Mohsen Kian$^1$, Trung Hoa Dinh$^2$, Mohammad Sal Moslehian$^3$, \MakeLowercase{and} Hiroyuki Osaka$^4$}

\address{$^1$Department of Mathematics, University of Bojnord, P. O. Box 1339, Bojnord 94531, Iran}
	\email{kian@ub.ac.ir }

\address{$^2$Department of Mathematics and Statistics, Troy University, Troy, AL 36082, USA}
\email{thdinh@troy.edu}

\address{$^3$Department of Pure Mathematics, Ferdowsi University of Mashhad, P. O. Box 1159, Mashhad 91775, Iran}
\email{moslehian@um.ac.ir}

\address{$^4$Department of Mathematical Sciences, Ritsumeikan University, Kusatsu, Shiga 525-8577, Japan}
\email{osaka@se.ritsumei.ac.jp}


 \begin{abstract}
Relating to finding possible upper bounds for the probability of error for discriminating between two
quantum states, it is well-known that
\begin{align*}
 \mathrm{tr}(A+B) - \mathrm{tr}|A-B|\leq 2\, \mathrm{tr}\big(f(A)g(B)\big)
\end{align*}
holds for every positive-valued matrix monotone function $f$, where $g(x)=x/f(x)$, and all positive definite matrices $A$ and $B$.

In this paper, we introduce a new class of functions that satisfy the above inequality. As a consequence, we derive a novel estimation of the quantum Chernoff bound. Additionally, we characterize matrix decreasing functions and establish matrix Powers-Störmer type inequalities for perspective functions.
\end{abstract}

\maketitle

\section{Introduction}
\vskip 3mm

Quantum hypothesis testing analyzes the optimality of measurement techniques designed to minimize the probability of error when distinguishing between two quantum states, denoted by $\rho$ and $\sigma$. Recall that a quantum state is represented by a positive semidefinite matrix with trace one. A key measure in this context is the \emph{quantum Chernoff bound} defined as
$$CH(\rho,\sigma)=\min\{\mathrm{tr}(\rho^s \sigma^{1-s});\, 0< s< 1\}.$$

The authors of \cite{Ad-} established an inequality relating the trace distance to the quantum Chernoff bound. Specifically, for positive definite matrices $A$ and $B$, they proved in \cite[Theorem 1]{Ad-} that for every $s\in[0,1]$,
\begin{align}\label{adhhh}
\mathrm{tr}(A+B) - \mathrm{tr}|A-B|\leq CH(A, B) \le 2\, \mathrm{tr}(A^s B^{1-s}).
\end{align} 
Notably, the left-hand side of (\ref{adhhh}) is independent of the function $x^s$. This observation suggests the possibility of refining (\ref{adhhh}) by extending its validity to a broader class of functions—beyond just power functions of the form $x^s$. Such an extension can lead to a generalized quantum Chernoff bound, offering new insights into quantum hypothesis testing.

Let $\mathcal{F}$ denote the class of matrix monotone functions $f: [0,\infty)\to[0,\infty)$ satisfying $f(0)>0$. For $A,B\in\mathbb{P}_n$, define
 $$CH_\mathcal{F}(A,B):=\inf\{\mathrm{tr}(f(A)g(B));\,\, f\in \mathcal{F},\, g(x)=x/f(x)\}.$$
Now, consider the subclass $ \mathcal{F}_0=\{f(x)=x^s;\, s\in[0,1]\}$. It follows that $$CH_{\mathcal{F}_0}(A,B)=CH(A,B).$$ This establishes a connection between the generalized definition $CH_\mathcal{F}(A,B)$ and the classical quantum Chernoff bound.

To explore potential refinements, the authors of \cite{HHH} replaced $A^s$ with $f(A)$ in \eqref{adhhh}, where $f$ is a general matrix monotone function defined on $[0,\infty)$ satisfying  $f(x)>0$ for all $x>0$. They \cite[Theorem 2.1]{HHH} established the following inequality:
\begin{align}\label{hhh}
\mathrm{tr}(A+B) - \mathrm{tr}|A-B|\leq 2 \mathrm{tr} (f(A)g(B)),
\end{align}
where $g(x)=x/f(x)$.
Consequently,
$$\frac{1}{2}\mathrm{tr}(A+B) - \phi(A,B)\leq CH_\mathcal{F}(A,B),$$
where $\phi(\rho,\sigma)=\|\rho-\sigma\|_1/2$ is the trace distance between two quantum states $\rho$ and $\sigma$.
In particular, 
\begin{align}\label{qch}
 1-CH(\rho,\sigma)\leq \phi(\rho,\sigma)\leq \sqrt{1-CH(\rho,\sigma)^2}
\end{align}
holds. The inequality
$$1 - \phi(\rho,\sigma)\leq CH_\mathcal{F}(\rho,\sigma)$$ may give a better estimation than \eqref{qch}; see \cite[Remark 3]{HHH}.

 Note that inequality \eqref{adhhh} is referred to as the \emph{Powers-St{\o}rmer's inequality} in the literature (see \cite[Section 2.4]{BKMS}). The second inequality in \eqref{qch} follows trivially from the relation $$\phi(\rho,\sigma)\leq \sqrt{1-(\mathrm{tr}(\rho^{1/2} \sigma^{1/2}))^2}$$ as shown in the arXiv version of \cite{Ad-}. It is also a direct consequence of \emph{Fuchs-van de Graaf’s inequality}: $$\phi(\rho,\sigma)\leq \sqrt{1-F(\rho,\sigma)^2},$$ where $F(\rho,\sigma)= \mathrm{tr}\left|\rho^{1/2} \sigma^{1/2}\right|$ is the so-called \emph{Fidelity}; see \cite{fv}. In fact, the right side of \eqref{qch} can be expressed more precisely as
\begin{align}\label{qchnew}
\phi(\rho,\sigma)\leq\sqrt{1-F(\rho,\sigma)^2}\leq\sqrt{1-(\mathrm{tr}(\rho^{1/2} \sigma^{1/2}))^2}\leq\sqrt{1-CH(\rho,\sigma)^2}.
\end{align}

In the literature, efforts have been made to extend this inequality to the setting of von Neumann algebras, as explored in \cite{jops}. It is shown in \cite[Theorem 1.12]{HH2} that \eqref{hhh} remains valid for the class $\mathcal C_{2n}$ of interpolation functions of order $2n$ provided that $f\circ g^{-1}$ is also a function in $\mathcal C_{2n}|_{g(0, \infty)}$. Notably, the class $\mathcal C_{2n}$ is broader than $\mathcal F$, which implies the inequality
$$
CH_{\mathcal C_{2n}}(A, B) \le CH_{\mathcal F}(A, B).
$$

In this paper, we aim to extend inequality \eqref{adhhh} to a new class of functions that goes beyond matrix monotone and matrix interpolation functions. This extension has the potential to yield a tighter upper bound in \eqref{hhh}. To support our claim, we provide several examples of such functions. In addition, we investigate corresponding matrix inequalities for this new class. As part of our exploration, we employ these inequalities to characterize matrix decreasing functions, which need not be positive.

\section{Preliminary}
Throughout the paper, let $\mathbb{M}_n$ denote the algebra of all $n\times n$ matrices with complex entries. A capital letter represents a matrix. If a matrix $A$ is positive semidefinite (positive definite, respectively), we write $A\geq 0$ ($A>0$, respectively). The well-known L\"{o}wner partial order is defined accordingly, that is, $A\leq B$ whenever $B-A\geq0$. We denote the set of all positive definite matrices in $\mathbb{M}_n$ as $\mathbb{P}_n$. The canonical trace mapping on $\mathbb{M}_n$ is denoted by $\mathrm{tr}(\cdot)$.

A continuous function $f: J \subseteq \mathbb{R}\to\mathbb{R} $ is said to be \emph{matrix monotone} (\emph{matrix decreasing}, respectively) if
$$A\leq B \quad \Longrightarrow \quad f(A)\leq f(B)\quad (f(A)\geq f(B), {\rm respectively})$$
holds for all Hermitian matrices $A$ and $B$ whose eigenvalues are contained in $J$. If for every $\lambda\in[0,1]$, we have $f(\lambda A+(1-\lambda)B)\leq \lambda f(A)+ (1-\lambda)f(B)$, then $f$ is called \emph{matrix convex}. The power function $f(t)=t^r$ is matrix monotone on $(0,\infty)$ if and only if $r\in [0,1]$. It is matrix convex on $(0,\infty)$ if $r\in [-1,0]\cup[1,2]$. If $r\in[-1,0]$, then the function $f$ is a matrix   decreasing function.

For every $A\in\mathbb{M}_n$, the positive semidefinite matrix $(A^*A)^{1/2}$ is called the \emph{absolute value} of $A$ and is denoted by $|A|$. Every Hermitian matrix $A$ has a \emph{Jordan decomposition}, $A=A_+-A_-$ in which $A_+=\frac{|A|+A}{2}$ and $A_-=\frac{|A|-A}{2}$ are positive semidefinite matrices and are called the \emph{positive part} and \emph{negative part} of $A$, respectively.

\vskip 3mm
\section{New estimation with new class of functions}
\vskip 3mm

Let $\widetilde{\mathcal{F}}$ be the class of continuous functions $f:[0,\infty)\to[0,\infty)$ for which $g^{-1}(x) = (x/f(x))^{-1}$ exists and $f\circ g^{-1}$ is matrix monotone on $g(0, \infty)$. It is straightforward to verify that for $s \in (1/2, 1),$ the function $f(x) = x^s$ is not in $\widetilde{\mathcal{F}}$, because $f\circ g^{-1}(x) = x^{s/(1-s)}$ is not an operator monotone function. At the same time, $x^s \in \widetilde{\mathcal{F}}$ for all $s \in [0, 1/2].$

\begin{proposition}
The following statements are true for the case of matrices of order $n \ge 2$:
\begin{itemize}
    \item[(i)] $\widetilde{\mathcal{F}} \setminus \mathcal{F} \neq \emptyset;$
    \item[(ii)] $\widetilde{\mathcal{F}} \setminus \mathcal{C}_{2n} \neq \emptyset.$
\end{itemize}
\end{proposition}
\begin{proof}
(i) It is sufficient to give an example of functions in $\widetilde{\mathcal{F}} \setminus \mathcal{F}$. Let us consider the function $f(x)=1/2(\sqrt{x(x+8)}-x)$. This function is continuous and positive-valued on $[0,\infty)$ but not $2$-monotone. Let us consider the functions $g(x)=\frac{x}{f(x)}=\frac{2x}{\sqrt{x(x+8)}-x}$ and $h(x)=\left(f\circ g^{-1}\right)(x)=\frac{2x}{x+1}$. The function $h$ is a known matrix monotone function on $[0, \infty)$, and so on $g(0, \infty) = (0, 2)$. Therefore, $f \in \widetilde{\mathcal{F}}.$

(ii) We need to show that the Lambert function belongs to $\widetilde{\mathcal{F}}$ with order $n=2$, and it is not an interpolation function of order $4.$ The \emph{Lambert function}, denoted by $W_k(x)$, also known as the \emph{omega function}, is the multi-valued (complex) inverse function of $f(z)=z\exp(z)$. If $x$ and $y$ are real numbers, then the equation $x=y\exp(y)$ has a solution (in terms of $y$) for every $x\geq -1/e$. For $-1/e\leq x<0$, there are two solutions, denoted by $W_0(x)$ and $W_{-1}(x)$. However, for every $x\geq0$, the solution is unique and is $W_0(x)$. The function $W_0$ is the principal branch of the Lambert function.

It is well-known that the exponential function $x\mapsto \exp(x)$ is not matrix monotone; see \cite[Example 1.8]{FMPS}. Consequently, we can find positive semidefinite matrices $A$ and $B$ such that $A\leq B$ but $\exp(A)\nleq \exp(B)$. Hence, it is enough to find positive definite matrices $X$ and $Y$ such that
\begin{align}\label{poe}
 0< X\leq Y \quad\mbox{and}\quad X\exp(X)\leq Y\exp(Y) \quad\mbox{and}\quad \exp(X)\nleq\exp(Y).
\end{align}
 By considering $A=X\exp(X)$ (i.e., $W_0(A)=X$) and $B=Y\exp(Y)$ (i.e., $W_0(B)=Y$) we can observe that $A\leq B$ but $\exp(W_0(A))=\exp(X)\nleq\exp(Y)=\exp(W_0(B))$.
Utilizing Matlab software, one can find matrices $X$ and $Y$ satisfying \eqref{poe}. For example, consider
\begin{align*}
X=\begin{bmatrix}
	10.98647 & 4.47043\\
	4.47043 & 5.18859
\end{bmatrix}
 \quad \mbox{and}\quad
  Y=\begin{bmatrix}
 	54.88824 & 48.34083 \\
 	48.340830 & 56.0955
 \end{bmatrix}.
\end{align*}
Since the class $\mathcal{C}_{2n}$ is a subset of the class of operator strictly monotone functions on $[0, \infty)$, from the above example one can see that the  Lambert function $W_0(x)$ is not an interpolation function of order $\mathcal C_4.$. Therefore, it is not in $\mathcal C_{2n}$ for any $n \ge 2.$

Finally, the Lambert function $f(x)=W_0(x)$ is in $\widetilde{\mathcal{F}}$. Since $f(x) \exp(f(x))=x$ for $x\geq0$, the matrix monotone function $h(x)=\log(x)$ satisfies $h\circ g=f$, where $g(x)=x/W_0(x)$. \end{proof}
\begin{remark}
 The Lambert function has numerous applications in quantum statistics \cite{valluri}, physics \cite{Istvan-Mezo}, biochemistry \cite{veberic}, and various other fields. It is used in solving differential equations, enumerating trees in combinatorics, and more. For further information on the Lambert function and its applications, we refer the reader to the book \cite{Istvan-Mezo}.

    The graphs of $f(x)=W_0(x)$ and $g(x)=x/W_0(x)$ are shown in Figure 1.
\begin{figure}[h!]\label{aks1}
	\begin{center}
	$
\begin{array}{cc}
 \includegraphics[height=5cm,width=5cm]{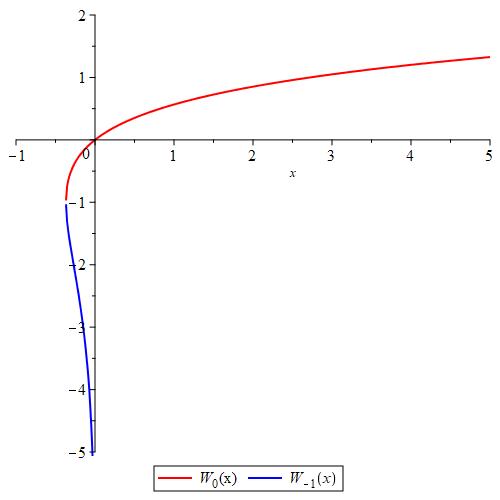} &
 \includegraphics[height=5cm,width=5cm]{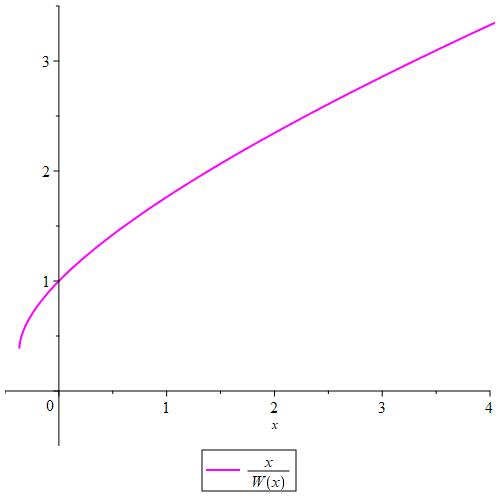}
 \end{array}
$\end{center}
\caption*{Graphs of $f(x)=W(x)$ (left) and $g(x)=x/W(x)$ (right)}
\end{figure}
\end{remark}

In the following theorem, we establish a Powers-Stormer type inequality for functions in the class $\widetilde{\mathcal{F}}$.

\begin{theorem}\label{thm}
Let $f$ be a continuous positive-valued function on $[0,\infty)$ and let $g(x)=x/f(x)$. If $g$ is invertible and $f\circ g^{-1}$ is matrix monotone, then the inequality
\begin{align}\label{as4}
\mathrm{tr} (A+B) -\mathrm{tr}|A-B| \leq 2\,\mathrm{tr} (g(A)f(B))
\end{align}
holds for all positive definite matrices $A$ and $B$.
\end{theorem}

As a consequence, from Theorem \ref{thm} and \eqref{hhh} we obtain the following inequalities, which give a better estimation in quantum hypothesis testing.
\begin{theorem}
With $CH_{\widetilde{\mathcal{F}}\cup\mathcal{F}}=\inf\{\mathrm{tr}\,(f(\rho)g(\sigma));\,\, f\in \mathcal{F}\cup\widetilde{\mathcal{F}},\, g(x)=x/f(x)\}$,  the double inequality
$$1-CH_{\widetilde{\mathcal{F}}\cup\mathcal{F}}(\rho,\sigma)\leq \phi(\rho,\sigma)\leq \sqrt{1-CH_{\widetilde{\mathcal{F}}\cup\mathcal{F}}(\rho,\sigma)^2}
$$
holds.
\end{theorem}

We now proceed to prove Theorem \ref{thm}. To this end, we require the following lemma, whose proof can be found in \cite[Corollary 1.5]{HH2} within the context of interpolation functions. Since the inequality in Lemma \ref{lm1} is of independent interest, we provide an alternative proof in the appendix, following the general approach outlined in \cite[Lemma 1]{Ad-}.

\begin{lemma}\label{lm1}
Let $f$ and $g$ be positive-valued matrix monotone and matrix   decreasing functions on $(0,\infty)$, respectively, and let $A, B \in\mathbb{P}_n$. If $P$ is the projection on the range of $(B-A)_+$, then
\begin{align}\label{as0}
\mathrm{tr}(PA(g(B)-g(A))) \geq 0
\end{align}
 and
 \begin{align}\label{as1}
 \mathrm{tr}(PA(f(A)-f(B))) \geq 0.
\end{align}
In particular,
$$\mathrm{tr}(PA(A^r-B^r)) \geq 0 \quad\mbox{and}\quad \mathrm{tr}(PA(B^s-A^s)) \geq 0 $$
holds for every $r\in[-1,0]$ and $s\in[0,1]$.
\end{lemma}

\textbf{Proof of Theorem \ref{thm}.} Suppose that the functions $g$ and $h$ are defined on $[0,\infty)$ as follows: $g(x)=x/f(x)$ and $h(x)=\left(f\circ g^{-1}\right)(x)$. Assume that $P$ is the projection onto the range of $(g(B)-g(A))_+$. An application of Lemma \ref{lm1} to the function $h$ and with $g(A)$ and $g(B)$ in place of $A$ and $B$, respectively, implies that
 $$\mathrm{tr}(Pg(A)(h(g(B))-h(g(A))))\geq 0.$$
 Noting that $h\circ g=f$ and subtracting both sides from $\mathrm{tr}(P(B-A))$ we conclude that
 $$\mathrm{tr}(P(B-A))-\mathrm{tr}(Pg(A)(f(B) -f(A)))\leq \mathrm{tr} (P(B-A)).$$
As $f(x)g(x)=x$, this gives
\begin{align}\label{pl1}
\mathrm{tr} (P(g(B)-g(A))f(B)) =\mathrm{tr}(P(B-g(A)f(B)))\leq \mathrm{tr}(P(B-A)).
\end{align}
 On the other hand, since $P$ is the projection on the range of $(g(B)-g(A))_+$, we have
{\small\begin{align}\label{pl2}
\mathrm{tr} (P(g(B)-g(A))f(B)) &= \mathrm{tr}\left((g(B)-g(A))_+\, f(B)\right)\nonumber\\
& \geq \mathrm{tr}((g(B)-g(A))f(B))=\mathrm{tr}(B-g(A)f(B)),
\end{align}}
where the inequality follows from the fact that $f(B)$ is a positive semidefinite matrix and $\mathrm{tr}(X_+Y_+)\geq \mathrm{tr}(XY_+)$ holds for all Hermitian matrices $X$ and $Y$. Moreover,
\begin{align}\label{pl3}
 \mathrm{tr}(P(B-A))\leq \mathrm{tr} ((B-A)_+)=\frac{1}{2}\,\left(\mathrm{tr} (B-A)+\mathrm{tr}|B-A|\right),
\end{align}
because $\mathrm{tr} (X_+) $ is the
maximum of $\mathrm{tr}(QX)$ over all projections $Q$.

 From \eqref{pl1} and \eqref{pl2} and \eqref{pl3} we can write
\begin{align*}
\mathrm{tr}B -\mathrm{tr}(g(A)f(B)) \leq 1/2\,\left(\mathrm{tr} (B-A)+\mathrm{tr}|B-A|\right),
\end{align*}
which is the desired inequality.

\section{Matrix Powers--St{\o}rmer's inequality for perspective functions and characterization of matrix decreasing functions}
\vskip 3mm

The authors of \cite{HKH} showed that the inequality
 \begin{align}\label{set}
 A+B-|A-B| \leq 2\, A \sigma_f B
 \end{align}
holds for all $A,B \in \mathbb{P}_n$ and all positive-valued matrix monotone functions $f$ on $[0,\infty)$ with $f(1)=1$, assuming that $AB+BA$ is positive definite. In inequality \eqref{set}, the joint monotonicity of the matrix mean $\sigma_f$ is crucial. Recall that the matrix mean $\sigma_f$ corresponding to the matrix monotone function $f$ is defined using the Kubo-Ando theory as follows:
\begin{align}\label{mmean}
\sigma_f(A,B)=A^{1/2}f(A^{-1/2}BA^{-1/2})A^{1/2},
\end{align}
where $A$ and $B$ are positive definite matrices.


The matrix perspective function $\mathcal{P}_f$  associated with a function $f$ is defined by a formula similar to that in \eqref{mmean}:
$$\mathcal{P}_f(A,B)=B^{1/2}f(B^{-1/2}AB^{-1/2})B^{1/2}.$$
It is known (see, for example, \cite{HSW}) that
\begin{align}\label{perses}
\mathcal{P}_f(X,Y)=\mathcal{P}_{\tilde{f}}(Y,X)\quad\mbox{and}\quad \mathcal{P}_f(X,Y)=\mathcal{P}_{f^*}(X^{-1},Y^{-1})^{-1}
\end{align}
where $\tilde{f}(x)=xf(x^{-1})$ and $f^*(x)=f(x^{-1})^{-1}$ are called the \emph{transpose function} and the \emph{adjoint function} of $f$, respectively. It is established that $f$ is matrix convex (concave) if and only if $\mathcal{P}_f$ is jointly matrix convex (concave); see \cite{HSW}.

However, if $f$ is matrix decreasing, the associated perspective function $\mathcal{P}_f$ may not be matrix decreasing in all of its variables. Specifically, while $\mathcal{P}_{f}$ is decreasing in the first variable, it is not necessarily decreasing in the second variable. For example if $f(x)=x^{-1}$, then $$\mathcal{P}_f (I,A)=A^2 \ngeqslant C^2=\mathcal{P}_f (I,C),$$ when $A\leq C$.

From the integral representation \eqref{intre} for matrix decreasing functions, we can express $\mathcal{P}_f(A,B)$ as
 \begin{align}\label{rert}
 \mathcal{P}_f(A,B)=\alpha B+\int_{0}^{+\infty}\frac{\lambda+1}{\lambda} ((\lambda BA^{-1}B):B)\, \mathrm{d}\nu(\lambda),
 \end{align}
in which $X:Y=(X^{-1}+Y^{-1})^{-1}$ represents the parallel sum of matrices $X$ and $Y$.
Since the parallel sum is jointly monotone, it follows that $\mathcal{P}_f$ is decreasing in the first variable.

In this section, we investigate a similar relation to \eqref{set} for perspective functions that are defined by a matrix   decreasing function. We also use such inequalities to characterize matrix decreasing functions.

Firstly, we provide a characterization of matrix decreasing functions under certain mild conditions.

\begin{theorem}\label{os}
Let $f:[0,\infty)\to(-\infty,\infty)$ be a continuous function.
\begin{itemize}
    \item[(i)]
 If $f$ is  matrix    decreasing with   $f(0)\leq 0$, then
 \begin{align}\label{sett 2}
 \mathcal{P}_f(A, B) \leq   f(1) (A \nabla_\alpha B - \max\{\alpha, 1- \alpha\}|A-B|)
\end{align}
 for all $A,B>0$ and $\alpha\in[0,1]$,  provided that $A \nabla_\alpha B - \max\{\alpha, 1- \alpha\}|A-B|$ is positive definite, and $A \nabla_\alpha B := (1 - \alpha)A + \alpha B$.

 \item[(ii)] If $f$ is matrix convex satisfying \eqref{sett 2}, then $f$ is matrix decreasing.
\end{itemize}
\end{theorem}
\begin{proof}
{\rm (i)} Assume that  $f$ is matrix decreasing on $[0, \infty)$ and put $g(x) = f(\frac{1}{x})$. Then $g$ is   matrix monotone on $(0, \infty)$. Hence, $\tilde{f}(x) = xg(x)$ is  matrix convex on $[0, \infty)$ by \cite[Theorem~1.13]{FMPS}.  Moreover, $\lim_{x\to\infty}\frac{\tilde{f}(x)}{x} = \lim_{x\to\infty}f(\frac{1}{x}) = f(0) \leq 0$. It follows from \cite[Theorem 7.2(2)]{HUW} that
\begin{align}\label{star}
A_1 \leq A_2 \quad\Rightarrow \quad \mathcal{P}_{\tilde{f}}(A_2, B) \leq \mathcal{P}_{\tilde{f}}(A_1, B)
\end{align}
for all $A_1, A_2, B >0$.
For every  $\alpha\in[0,1]$ we can write
\begin{align*}
A &= (1 - \alpha)A + \alpha B + \alpha (A - B)\\
&\geq A \nabla_\alpha B - \alpha |A - B|\\
&\geq A\nabla_\alpha B - \alpha_0 |A-B|,
\end{align*}
where $\alpha_0=\max\{\alpha, 1- \alpha\}$. Similarly, we have $B \geq A\nabla_\alpha B - \alpha_0|A-B|$.

If the matrix   $A\nabla_\alpha B - \alpha_0 |A-B|$  is positive,  then
\begin{align*}
\mathcal{P}_f(A, B) &\leq \mathcal{P}_f(A\nabla_\alpha B - \alpha_0|A-B|, B)\qquad (\mbox{by \eqref{star}})\\
&= \mathcal{P}_{\tilde{f}}(B, A\nabla_\alpha B - \alpha_0|A-B|) \qquad (\mbox{by \eqref{perses}})\\
&\leq \mathcal{P}_{\tilde{f}}(A\nabla_\alpha B - \alpha_0|A-B|, A\nabla_\alpha B - \alpha_0|A-B|) \qquad (\mbox{by \eqref{star}})\\
&=  f(1) (A\nabla_\alpha B - \alpha_0|A-B|),
\end{align*}
in which we use $f(1)=\tilde{f}(1)$.

{\rm (ii)} Assume that $f$ is matrix convex and \eqref{sett 2} holds. Put $A=t I$  and $B=I$. If $t<1$, then $f(t)\leq f(1)t$ and so $ \lim\sup_{t\to 0+}f(t) \leq 0 $. Moreover, employing \eqref{sett 2} with the same $A$ and $B$, but for $t>1$ we arrive at
 $f(1) \geq f(t)$, whence $\frac{f(t)}{t} \leq \frac{f(1)}{t}$. This implies that $\lim_{t\to \infty}\frac{f(t)}{t} \leq 0$.
It follows from \cite[Corollary~2.7]{OW}  and \cite[Proposition 6.1(iii)]{HSW} that there exist $\lambda, \mu \in \mathbb{R}$ and a matrix connection $\sigma$ such that
$$
\mathcal{P}_f(A, B) = \lambda B + \mu A - A\sigma B
$$
for all $A, B > 0$, where $ \lambda = \lim\sup_{t\to 0^+} f(t)$ and $ \mu = \lim\sup_{t\to\infty}\frac{f(t)}{t}$. Consequently,
$
\mathcal{P}_f(A, B) $
is   decreasing in both variables.
Accordingly, from $f(A) = \mathcal{P}_f (A,I) = \lambda I + \mu A- I\sigma A$, we infer that $f$ is matrix   decreasing.
\end{proof}

\begin{remark}
  We note that Theorem~\ref{os} {\rm (i)} generalizes \cite[Theorem 2.1]{HKH}. Indeed, if $g$ is a positive matrix monotone function with $g(1)=1$, then we can apply Theorem~\ref{os} to the matrix   decreasing function $f=-g$ to get
  \begin{align*}
B\sigma_g A= \mathcal{P}_g(A, B) \geq      A \nabla_\alpha B - \max\{\alpha, 1- \alpha\}|A-B|.
\end{align*}
Letting $\alpha=1/2$, we reach \eqref{set}.

\end{remark}

\begin{corollary}\label{thqw}
 Let $f:[0,\infty)\to\mathbb{R} $ be a matrix    decreasing   function   with $f(0)\leq 0$. Then
 \begin{align}\label{settt}
\frac{1}{2}f(1)(A+B-|A-B|) \geq   \mathcal{P}_f(A,B)
 \end{align}
 holds for all $A,B>0$, provided that $AB+BA$ is a positive definite matrix.
 \end{corollary}
\begin{proof}
Take $\alpha = \frac{1}{2}$ in Theorem~\ref{os}.
Indeed, since $AB + BA \geq 0$,
\begin{align*}
A\nabla_\alpha B - \max\{\alpha, 1-\alpha\}|A-B|
&= \frac{A+B}{2} - \frac{|A-B|}{2}\\
&= \frac{(A^2+B^2 +(AB + BA))^{\frac{1}{2}}}{2}-\frac{|A-B|}{2}\\
&\geq \frac{(A^2 + B^2 - (AB + BA))^{\frac{1}{2}}}{2} - \frac{|A-B|}{2}\\
&= \frac{|A-B|}{2} - \frac{|A-B|}{2} =0
\end{align*}
The result then follows from Theorem~\ref{os} {\rm (i)}.
\end{proof}

\begin{remark}
Let $f:[0, \infty) \to (- \infty, \infty)$ be a continuous function. If $f$ is a matrix   decreasing function, then  $g(x) = f(x) - f(0)$
  is also matrix   decreasing and $g(0) = 0$. For all positive definite matrices $A$ and $B$, there exists a positive real number $t$ such that $0<t I\leq A$ and $0<t I\leq B$.  Hence
  \begin{align*}
    \mathcal{P}_g(A,B)&\leq \mathcal{P}_g(tI,B)=\mathcal{P}_{\tilde{g}} (B,tI) \leq  \mathcal{P}_{\tilde{g}} (tI,tI)=t g(1)I=t(f(1)-f(0))I\leq 0.
  \end{align*}
Consequently,  for any matrix   decreasing function $f:[0, \infty) \to (- \infty, \infty)$ we have
$$\mathcal{P}_f(A,B)=\mathcal{P}_g(A,B)+f(0)B\leq f(0)B$$
 for all positive definite matrices $A$ and $B$.
\end{remark}

In the following theorem we establish a lower bound for $\mathcal{P}_f(A, B)$.

\begin{theorem}\label{prp}
  Let $f: [0, \infty) \to \mathbb{R}$ be   matrix   decreasing with $f(0) \leq 0$. Then
 \begin{align*}
\mathcal{P}_f(A, B) \geq f(1)(A \nabla_\alpha B + \max\{\alpha, 1- \alpha\}|A-B|)
\end{align*}
 for all $A,B>0$ and $\alpha\in[0,1]$.
\end{theorem}
\begin{proof}
As in the proof of Theorem~\ref{os} {\rm (i)}, we have $\lim_{x\to\infty}\frac{\tilde{f}(x)}{x} = \lim_{x\to\infty}f(\frac{1}{x}) = f(0) \leq 0$ and so
\begin{align}\label{star2}
A_1 \leq A_2 \quad\Rightarrow \quad \mathcal{P}_{\tilde{f}}(A_2, B) \leq \mathcal{P}_{\tilde{f}}(A_1, B)
\end{align}
for all $A_1, A_2, B >0$. Moreover,
\begin{align*}
A &= (1 - \alpha)A + \alpha B + \alpha (A - B)\leq A \nabla_\alpha B + \alpha |A - B|\leq A\nabla_\alpha B + \alpha_0 |A-B|
\end{align*}
for every  $\alpha\in[0,1]$, where $\alpha_0=\max\{\alpha, 1- \alpha\}$. Similarly, we have $B \leq A\nabla_\alpha B + \alpha_0|A-B|$.

Therefore,
\begin{align*}
\mathcal{P}_f(A, B) &\geq \mathcal{P}_f(A\nabla_\alpha B + \alpha_0|A-B|, B)\qquad (\mbox{by \eqref{star2}})\\
&= \mathcal{P}_{\tilde{f}}(B, A\nabla_\alpha B + \alpha_0|A-B|) \qquad (\mbox{by \eqref{perses}})\\
&\geq \mathcal{P}_{\tilde{f}}(A\nabla_\alpha B + \alpha_0|A-B|, A\nabla_\alpha B + \alpha_0|A-B|) \qquad (\mbox{by \eqref{star2}})\\
 &\geq f(1)(A\nabla_\alpha B + \alpha_0|A-B|) \  (\tilde{f}(1) = f(1) )
\end{align*}
in which we use $f(1)=\tilde{f}(1)$. \\
\end{proof}

 \begin{remark}
 Let $f$ be a positive matrix monotone function on $[0,\infty)$. Then $-f$ satisfies the conditions of Proposition \ref{prp} and so
  \begin{align*}
B\sigma_f A=\mathcal{P}_f(A, B) \leq f(1)(A \nabla_\alpha B + \max\{\alpha, 1- \alpha\}|A-B|)
\end{align*}
 for all $A,B>0$ and $\alpha\in[0,1]$. In particular, $f(1)=1$ and $\alpha=1/2$ entails
   \begin{align*}
2B\sigma_f A  \leq   A + B +|A-B|
\end{align*}
for all $A,B>0$, which is \cite[Proposition 2.3]{HKH}.
 \end{remark}


\bigskip

\section{Appendix: A new proof of Lemma \ref{lm1} }

\begin{proof}

To prove \eqref{as1}, we employ an integral representation for matrix   decreasing functions due to Hansen  \cite{Ha}  as
\begin{align}\label{intre}
f(x)=\alpha+\int_{0}^{+\infty}\frac{\lambda + 1}{\lambda+x}\,\mathrm{d}\nu(\lambda),
\end{align}
where $\alpha\geq0$ and $\nu$ is a finite positive measure on $[0,\infty)$.

If $\varphi:[0,1]\to\mathbb{P}_n$ is a differentiable map, then
\begin{align}\label{dif}
\frac{\mathrm{d}}{\mathrm{d}y}\left(-\varphi(y)^{-1}\right)=\varphi(y)^{-1}
\left(\frac{\mathrm{d}}{\mathrm{d}y}\varphi(y)\right)\,\varphi(y)^{-1}
\end{align}
and
\begin{align*}
\varphi(0)^{-1}-\varphi(1)^{-1}=\int_{0}^{1}\mathrm{d}y \frac{\mathrm{d}}{\mathrm{d}y}\left(-\varphi(y)^{-1}\right).
\end{align*}
 By considering $\varphi(y)=T+(S-T)y$, the identity
\begin{align}\label{intinv}
T^{-1}-S^{-1}=\int_{0}^{1}\mathrm{d}y \frac{\mathrm{d}}{\mathrm{d}y}(-T-(S-T)y)^{-1}
\end{align}
is valid for all invertible matrices $T$ and $S$. In particular, for any positive real number $\lambda$, the matrices $T=A+\lambda$ and $S=B+\lambda$ are positive definite and we have
\begin{align}\label{wq}
 (A+\lambda)^{-1}-(B+\lambda)^{-1}=\int_{0}^{1}\mathrm{d}y \frac{\mathrm{d}}{\mathrm{d}y}\left(-(A+\lambda+(B-A)y)^{-1}\right).
\end{align}
Moreover, applying \eqref{dif} for a proper function $\varphi$, we have
\begin{align}\label{ner}
\begin{split}
&\frac{\mathrm{d}}{\mathrm{d}y}\left(-(A+\lambda+(B-A)y)^{-1}\right)\\
&=
(A+\lambda+(B-A)y)^{-1}\frac{\mathrm{d}}{\mathrm{d}y}(A+\lambda+(B-A)y)(A+\lambda+(B-A)y)^{-1}\\
&=(A+\lambda+(B-A)y)^{-1}(B-A)(A+\lambda+(B-A)y)^{-1}\\
&=y^{-1}\, Z N Z,
\end{split}
\end{align}
where
\begin{align}\label{qwb}
 N= (B-A)y \qquad \mbox{and}\qquad Z=(A+\lambda+N)^{-1}.
 \end{align}
 Applying the integral representation \eqref{intre}, we can obtain
\begin{align*}
f(A)-f(B)&=\int_{0}^{+\infty}(\lambda+1)\left((A+\lambda)^{-1}-(B+\lambda)^{-1}\right)\,\mathrm{d}\nu(\lambda)\\
&= \int_{0}^{+\infty}(\lambda+1)\,\mathrm{d}\nu(\lambda)\int_{0}^{1} y^{-1}Z N Z \mathrm{d}y,
 \end{align*}
 in which the second equality follows from \eqref{wq} and \eqref{ner}.
 Hence, to prove \eqref{as1}, it is enough to show that $\mathrm{tr} (PAZ N Z)$ is positive. Let $N=N_+ -N_-$ represent the Jordan decomposition of $N$, such that $P$ is the projection onto the range of $N_+$. Decomposing $\mathbb{C}^n$ as the direct sum of the range of $N_+$ and its orthogonal complement, one can express $N$ and $P$ as
\begin{align*}
 N= \begin{bmatrix}
N_+ & 0 \\
0 & -N_-
 \end{bmatrix}
 \quad
 \mbox{and}
 \quad
 P= \begin{bmatrix}
I & 0 \\
0 & 0
 \end{bmatrix};
\end{align*}
see \cite[Section 1.2]{MO}. It follows from \eqref{qwb} that $A=Z^{-1}-N-\lambda$, and hence
\begin{align}\label{er2}
 \mathrm{tr} (PAZNZ) & =\mathrm{tr}\big(P((Z^{-1}-N)ZNZ-\lambda ZNZ)\big)\nonumber \\
 &=\mathrm{tr}\big(PN(Z-ZNZ) -\lambda PZNZ\big)\nonumber\\
 &=\mathrm{tr}\big(N_+(Z-ZNZ)-\lambda PZNZ\big).
\end{align}
 Moreover, since $A$ is positive definite, we have
\begin{align*}
 0\leq ZAZ=Z(Z^{-1}-N-\lambda)Z=Z-ZNZ-\lambda Z^2,
\end{align*}
whence
\begin{align}\label{new21}
\lambda Z^2 \leq Z-ZNZ.
\end{align}
 Finally,
 \begin{align*}
 \mathrm{tr} (PAZNZ) & =\mathrm{tr}\big(N_+(Z-ZNZ)-\lambda PZNZ\big)\quad\qquad(\mbox{by \eqref{er2}})\\
 &\geq \lambda\Big(\mathrm{tr}\big(N_+Z^2\big)-\mathrm{tr}(PZNZ)\Big)\quad\qquad\qquad\quad(\mbox{by \eqref{new21}})\\
 & = \lambda\Big(\mathrm{tr} (ZN_+Z)-\mathrm{tr}(PZNZ)\Big),\\
 & \geq \lambda\Big(\mathrm{tr} (PZN_+Z)-\mathrm{tr}(PZNZ)\Big).
 \end{align*}
The last inequality follows from the fact that $P$ is a projection. The last term on the right is positive as $N_+\geq N$, which completes the proof of \eqref{as1}.

For \eqref{as0}, we use the well-known integral representation for matrix monotone functions
\begin{align}\label{intre2}
g(x)=\alpha+\beta x+ \int_{0}^{+\infty}\frac{x}{x+\lambda}\,\mathrm{d}\mu(\lambda),
\end{align}
with positive scalars $\alpha,\beta$ and positive Borel measure $\mu$ on $(0,\infty)$, as shown in \cite{Bh}. Accordingly, we have
$$PA(g(B)-g(A))=\beta PA(B-A)+PA\int_{0}^{+\infty}\left(\frac{B}{\lambda+B}-\frac{A}{\lambda+A}\right)\,\mathrm{d}\mu(\lambda).$$
For every $\lambda>0$, we have
\begin{align*}
 \frac{B}{\lambda+B}-\frac{A}{\lambda+A}=\lambda(A+\lambda)^{-1}-\lambda(B+\lambda)^{-1}.
\end{align*}
 Therefore, by applying the same argument as in the first part of the proof, we can demonstrate that
 $$\mathrm{tr}\left( PA\left(\frac{B}{\lambda +B}-\frac{A}{\lambda +A}\right)\right) \geq0 $$
for every $\lambda>0$. Therefore, it is enough to prove that
$ \mathrm{tr}(PA(B-A))\geq0$.
To do this, we are aware that $B-A$ is Hermitian, so $(B-A)P=P(B-A)=(B-A)_+\geq 0$. Hence,
\begin{align*}
\mathrm{tr}(PA(B-A)) = \mathrm{tr}(A(B-A)P) = \mathrm{tr}(A(B-A)_+) = \mathrm{tr}(A^{1/2}(B-A)_+A^{1/2})\geq 0.
\end{align*}
 Thus, the proof of \eqref{as0} is now completed.
\end{proof}

\medskip


\noindent \textbf{Acknowledgment.} The forth author's was supported by the JSPS grant for Scientific Research No. 20K03644.

\medskip

\noindent \textit{Conflict of Interest Statement.} On behalf of all authors, the corresponding author states that there is no conflict of interest. 

\medskip

\noindent\textit{Data Availability Statement.} Data sharing is not applicable to this article as no datasets were generated or analyzed during the current study.

\medskip

\end{document}